\documentclass{llncs}

\usepackage{multicol}

\usepackage{amsmath}
\usepackage{amsfonts}
\usepackage{amssymb}

\usepackage{graphicx}
\usepackage{epic}
\usepackage{eepic}
\usepackage{epsfig,float}
\usepackage{verbatim}
\usepackage{pdfsync}

\renewcommand{\le}{\leqslant}
\renewcommand{\ge}{\geqslant}
\newcommand{\cmp}{\overline}
\newcommand{\ol}{\overline}
\newcommand{\eps}{\varepsilon}
\newcommand{\emp}{\emptyset}

\newcommand{\Sig}{\Sigma}
\newcommand{\sig}{\sigma}
\newcommand{\noin}{\noindent}
\newcommand{\defeq}{\stackrel{\rm def}{=}}

\newcommand{\bi}{\begin{itemize}}
\newcommand{\ei}{\end{itemize}}
\newcommand{\be}{\begin{enumerate}}
\newcommand{\ee}{\end{enumerate}}
\newcommand{\bd}{\begin{description}}
\newcommand{\ed}{\end{description}}
\newcommand{\bq}{\begin{quote}}
\newcommand{\eq}{\end{quote}}
\newcommand{\txt}[1]{\mbox{ #1 }}

\newcommand{\cA}{{\mathcal A}}
\newcommand{\cB}{{\mathcal B}}
\newcommand{\cC}{{\mathcal C}}

\newcommand{\cI}{{\mathcal I}}

\newcommand{\cT}{{\mathcal T}}

\newcommand{\raL}{{\hspace{.1cm}{\sim_L} \hspace{.1cm}}}
\newcommand{\lraL}{{\hspace{.1cm}{\approx_L} \hspace{.1cm}}}

\newcommand{\tr}{{transformation}}

\newcommand{\timg}{\mbox{img}}
\newcommand{\tdom}{\mbox{dom}}

\title{Syntactic Complexity of Star-Free Languages\thanks{This work was supported by the Natural Sciences and Engineering Research Council of Canada under grant No.~OGP0000871.
}
}

\author{Janusz~Brzozowski and Baiyu Li
 }

\authorrunning{Brzozowski, Li}   

\institute{David R. Cheriton School of Computer Science\\ University of Waterloo, 
Waterloo, ON, Canada N2L 3G1\\
\{{\tt \{brzozo,b5li\}@uwaterloo.ca}\}
}

\begin{document}

\maketitle
\today
\begin{abstract}
The syntactic complexity of a regular language is the cardinality of its syntactic semigroup.
The syntactic complexity of a subclass of regular languages is the maximal syntactic complexity of languages in that subclass, taken as a function of the state complexity of these languages.
We study the syntactic complexity of  star-free regular languages, that is, languages that can be constructed from finite languages using union, complement and concatenation. 
We find tight upper bounds on the syntactic complexity of languages accepted by monotonic and partially monotonic automata.
We introduce ``nearly monotonic'' automata, which accept star-free languages, and find a tight upper bound on the syntactic complexity of languages accepted by such automata.
We conjecture that this bound is also an upper bound on the syntactic complexity of star-free languages.
\bigskip

\noin
\end{abstract}

\section{Introduction}

The class of regular languages is the smallest class containing the finite languages and closed under union, concatenation and star. Since regular languages are also closed under complementation, one can redefine them as the smallest class containing finite languages and closed under boolean operations, concatenation and star. 
\emph{Star-free} languages constitute the smallest class containing the finite languages and closed under boolean operations and concatenation. 
In 1965, Sch\"utzenberger proved~\cite{Sch65} that a language is star-free if and only if its syntactic monoid is \emph{group-free}, that is, has only trivial subgroups. An equivalent condition is that the minimal deterministic automaton of a star-free language is \emph{permutation-free}, that is, has only trivial permutations. Another point of view is that these automata are \emph{counter-free}, since they cannot count modulo any integer greater than 1. They can, however, \emph{count to a threshold}, that is $0,1, \ldots, (n-1$ \emph{or more}).
Such automata are called \emph{aperiodic,} and this is the term  we use. Star-free languages were  studied in detail in 1971 by McNaughton and Papert~\cite{McNP71}. 

The \emph{state complexity of a regular language} is the  number of states in the minimal deterministic finite automaton (DFA) recognizing that language. 
The \emph{state complexity of an operation} on languages in a subclass of regular languages is the maximal  state complexity of the language resulting from the operation taken as a function of the state complexities of the arguments in that subclass.
State complexity of regular operations has been studied quite extensively; for a survey of this topic and a list of references see~\cite{Yu01}. 

An equivalent notion to state complexity is that of \emph{quotient complexity}~\cite{Brz10}, which is the number of left quotients of the language; we prefer to use the term quotient complexity.

The notion of quotient complexity can be derived from the Nerode equivalence~\cite{Ner58}, 
while the Myhill equivalence~\cite{Myh57} leads to the syntactic semigroup of a language and to its \emph{syntactic complexity}, which is the cardinality of the syntactic semigroup.
It was pointed out in~\cite{BrYe10} that syntactic complexity can be very different for regular languages with the same quotient complexity.

In contrast to state complexity, syntactic complexity has not received much attention. In 1970 Maslov~\cite{Mas70} dealt with generators of the semigroup of all transformations of a finite set in the setting of finite automata. In 2003--2004, Holzer and K\"onig~\cite{HoKo04}, and independently, Krawetz, Lawrence and Shallit~\cite{KLS03} studied the syntactic complexity of automata with unary and binary alphabets. 
In 2010 Brzozowski and Ye~\cite{BrYe10} examined ideal and closed languages, and in 2011 Brzozowski, Li and Ye~\cite{BLY11} studied  prefix-, suffix-, and bifix-free languages and their complements.

Here we deal with star-free languages.
It has been shown in 2011 by Brzozowski and Liu~\cite{BrLiu11} that boolean operations, concatenation, star, and reversal in the class of star-free  languages meet all the quotient complexity bounds of regular languages, with a few exceptions. Also, Kutrib, Holzer, and 
Meckel~\cite{HKM11} proved in 2011 that in most cases exactly the same tight state complexity bounds  are reached by operations on aperiodic nondeterministic finite automata (NFA's) as on general NFA's.
In sharp contrast to this, we show that the syntactic complexity of star-free languages is much smaller than that of regular languages, which is $n^n$ for languages with quotient complexity $n$.
We derive tight upper bounds for three subclasses of star-free languages, the monotonic, partially monotonic, and nearly monotonic languages.
We conjecture that the bound for star-free languages is the same as  for 
nearly monotonic languages.

The remainder of the paper is structured as follows. Our terminology and some basic facts are stated in Section~\ref{sec:trans}. Aperiodic transformations are examined in Section~\ref{sec:aper}. In Section~\ref{sec:mono}, we study monotonic,  partially monotonic, and nearly monotonic automata. In Section~\ref{sec:sc3} we derive a tight upper bound on the syntactic complexity of star-free languages having quotient complexity $3$. Section~\ref{sec:con} concludes the paper.

\section{Preliminaries}\label{sec:trans}

We assume that the reader is familiar with basic theory of formal languages as described in~\cite{RozSal97}, for example. Let $\Sig$ be a non-empty finite alphabet and $\Sig^*$, the free monoid generated by $\Sig$. A \emph{word} is any element of $\Sig^*$, and the empty word is $\eps$. The length of a word $w\in \Sig^*$ is $|w|$. 

A \emph{language} over $\Sig$ is any subset of $\Sig^*$. 
For any languages $K$ and $L$ over $\Sig$, we consider the following {\em boolean operations}: complement ($\cmp{L} = \Sig^* \setminus L$), union ($K \cup L$), intersection ($K \cap L$), difference ($K \setminus L$), and symmetric difference ($K \oplus L$). The {\em product}, or \emph{(con)catenation}, of $K$ and $L$ is $KL = \{w \in \Sig^* \mid w = uv, u \in K, v \in L\}$; and the {\em star} of $L$ is $L^* = \bigcup_{i \ge 0}L^i$. 

Regular languages are defined by induction: For the basis, the languages  $\emptyset$, $\{\eps\}$, and $\{a\}$ for any $a \in \Sig$ are regular. Then if $K$ and $L$ are regular,  
so are $K\cup L$, $KL$ and $L^*$. Finally, no language is regular unless it is constructed from the basic languages using  a finite number of the operations of union, product and star. Since regular languages are also closed under complements, we can redefine them as the smallest class containing the basic languages and closed under boolean operations, product, and star.
{\em Star-free} languages are the smallest class of languages constructed from the basic languages using only boolean operations and product.

A~\emph{deterministic finite automaton} (DFA) is a quintuple $\cA=(Q, \Sig, \delta, q_1,F)$, where 
$Q$ is a finite, non-empty set of \emph{states}, $\Sig$ is a finite non-empty \emph{alphabet}, $\delta:Q\times \Sig\to Q$ is the \emph{transition function}, $q_1\in Q$ is the \emph{initial state}, and $F\subseteq Q$ is the set of \emph{final states}. We extend $\delta$ to $Q \times \Sig^*$ in the usual way.
The DFA $\cA$ accepts a word $w \in \Sigma^*$ if ${\delta}(q_1,w)\in F$. 
The set of all words {\it accepted} by $\cA$ is $L(\cA)$. 
Regular languages are exactly the languages accepted by DFA's. 
By the \emph{language of a state} $q$ of $\cA$ 
we mean the language $L_q$ accepted 
by the DFA $(Q,\Sigma,\delta,q,F)$. 
A state is \emph{empty} if its language is empty.

An \emph{incomplete deterministic finite automaton (IDFA)} is a quintuple 
$\cI=(Q, \Sig, \delta, q_1,F)$, where $Q$, $\Sig$, $q_1$ and $F$ are as in a DFA, and $\delta$ is a partial function such that either 
$\delta(q,a)=p$ for some $p\in Q$ or $\delta(q,a)$ is undefined for any $p,q\in Q$, $a\in \Sig$.
Every DFA is also an IDFA.

The \emph{left quotient}, or simply \emph{quotient,} of a language $L$ by a word $w$ is the language $L_w=\{x \in \Sig^*\mid wx\in L \}$. 
For any language $L$ over $\Sig$, the \emph{Nerode equivalence} $\raL$ of $L$ is defined as follows~\cite{Ner58}: 
\begin{equation*}
x \raL y \mbox{ if and only if } xv\in L  \Leftrightarrow yv\in L, \mbox { for all } v\in\Sig^*.
\end{equation*}
Clearly, $L_x=L_y$ if and only if $x\raL y$.
Thus each equivalence class of the Nerode equivalence corresponds to a distinct quotient of $L$.

Let $L$ be a regular language. 
The \emph{quotient DFA} of $L$ is 
$\cA=(Q, \Sig, \delta, q_1,F)$, where $Q=\{L_w\mid w\in\Sig^*\}$, $\delta(L_w,a)=L_{wa}$, 
$q_1=L_\eps=L$,  and $F=\{L_w \mid \eps \in L_w\}$.
Note that every state $L_w$ of a quotient DFA is reachable from the initial state $L$ by the word $w$. Also, the language of every state is distinct, since only distinct quotients are used as states. Thus every quotient DFA is minimal.
The \emph{quotient IDFA} of $L$ is the quotient DFA of $L$ after the empty state (if present) and all transitions incident to it are removed. 
The quotient IDFA is also minimal. 

If a regular language $L$ has quotient IDFA $\cI$, then the DFA $\cA$ obtained by adding an empty quotient to $\cI$, if necessary, is the quotient DFA of $L$.
Conversely, if $L$ has quotient DFA $\cA$, then the IDFA $\cI$ obtained from $\cA$ by removing the empty quotient, if present, is the quotient IDFA of $L$.
The two automata $\cA$ and $\cI$ are equivalent, in the sense that they accept the same language.

The number  $\kappa(L)$ of distinct quotients of $L$ is the \emph{quotient complexity} 
of~$L$. 
Since the quotient DFA of $L$ is minimal, quotient complexity is the same as  state complexity, but there are advantages to using quotients~\cite{Brz10}.

The \emph{Myhill equivalence} $\lraL$ of $L$ is defined as follows~\cite{Myh57}:
\begin{equation*}
x \lraL y \mbox{ if and only if } uxv\in L  \Leftrightarrow uyv\in L\mbox { for all } u,v\in\Sig^*.
\end{equation*}
This equivalence is also known as the \emph{syntactic congruence} of $L$.
The quotient set $\Sig^+/ \lraL$ of equivalence classes of the relation $\lraL$,  is a semigroup called the \emph{syntactic semigroup} of $L$ (which we denote by $T_L$), and 
$\Sig^*/ \lraL$ is the \emph{syntactic monoid} of~$L$. 
The \emph{syntactic complexity} $\sig(L)$ of $L$ is the cardinality of its syntactic semigroup.
The \emph{monoid complexity} $\mu(L)$ of $L$ is the cardinality of its syntactic monoid.
If the  equivalence class containing $\eps$ is a singleton in the syntactic monoid, then $\sig(L)=\mu(L)-1$; otherwise, $\sig(L)=\mu(L)$.

\smallskip

A {\em partial transformation} of a set $Q$ is a partial mapping of $Q$ into itself; we consider partial transformations of finite sets only, and we assume without loss of generality  that $Q = \{1,2,\ldots, n\}$. Let $t$ be a partial transformation of $Q$. If $t$ is defined for $i \in Q$, then $it$ is the image of $i$ under $t$; otherwise it is undefined and we write $it = \Box$. For convenience, we let $\Box t = \Box$. If $X$ is a subset of $Q$, then $Xt = \{it \mid i \in X\}$. The {\em composition} of two partial transformations $t_1$ and $t_2$ of $Q$ is a partial transformation $t_1 \circ t_2$ such that $i (t_1 \circ t_2) = (i t_1) t_2$ for all $i \in Q$. We usually drop the composition operator ``$\circ$'' and write $t_1t_2$ for short. 

An arbitrary partial transformation can be written in the form
\begin{equation*}\label{eq:transmatrix}
t=\left( \begin{array}{ccccc}
1 & 2 &   \cdots &  n-1 & n \\
i_1 & i_2 &   \cdots &  i_{n-1} & i_n
\end{array} \right ),
\end{equation*}
where $i_k = kt$ and $i_k\in Q \cup \{\Box\}$, for $1\le k\le n$. The {\em domain} of $t$ is the set 
$ \tdom~t = \{k \in Q \mid kt \neq \Box\}.$ 
The {\em image} of $t$ is the set
$\timg~t = (\tdom~t) t = \{kt \mid k \in Q \txt{and} kt \neq \Box\}.$ We also use the notation $t = [i_1,\ldots,i_n]$ for the partial transformation $t$ above. 

A  \emph{(full) transformation} of a set $Q$ is a total mapping of $Q$ \emph{into} itself, whereas a \emph{permutation} of $Q$ is a mapping of $Q$ \emph{onto} itself. In other words, a transformation $t$ of $Q$ is a partial transformation where $\tdom~t = Q$, and a permutation of $Q$ is a transformation $t$ where $\timg~t = Q$. 

The \emph{identity} transformation maps each element to itself, that is, $it=i$ for $i=1,\ldots,n$.
A transformation $t$ contains a \emph{cycle} of length $k$ if there exist pairwise distinct elements $i_1,\ldots,i_k$ such that
$i_1t=i_2, i_2t=i_3,\ldots, i_{k-1}t=i_k$, and $i_kt=i_1$.
A cycle is denoted by $(i_1,i_2,\ldots,i_k)$.
For $i<j$, a \emph{transposition} is the cycle $(i,j)$, and $(i,i)$ is the identity.
A \emph{singular} transformation, denoted by $i\choose j$, has $it=j$ and $ht=h$ for all $h\neq i$, and $i \choose i$ is the identity.
A~\emph{constant} transformation,  denoted by $Q \choose j$, has $it=j$ for all $i$.
Let $\cT_Q$ be the set of all full transformations of $Q$.

In this paper we study full transformations performed by DFA's and partial transformations performed by IDFA's. 

Let $\cA = (Q, \Sig, \delta, q_1, F)$ be a DFA. For each word $w \in \Sig^*$, the transition function defines a full transformation $t_w$ of $Q$: for all $i \in Q$, 
$it_w \defeq \delta(i, w).$ 
The set $T_{\cA}$ of all such transformations by non-empty words forms a subsemigroup of $\cT_Q$, called the \emph{transition semigroup} of $\cA$~\cite{RozSal97}. 
Conversely, we can use a set  $\{t_a \mid a \in \Sig\}$ of transformations to define $\delta$, and so the DFA $\cA$. When the context is clear we simply write $a = t$, where $t$ is a transformation of $Q$, to mean that the transformation performed by $a \in \Sig$ is~$t$. If $\cA$ is the quotient DFA of $L$, then $T_{\cA}$ is isomorphic to the syntactic semigroup $T_L$ of $L$ \cite{McNP71}, and we represent elements of $T_L$ by transformations in $T_{\cA}$. 

If $\cI$ is an IDFA, then each word $w\in\Sig^*$ performs a partial transformation of~$Q$.
The set of all these partial transformations is the \emph{transition semigroup}  of~$\cI$.
If $\cI$ is the quotient IDFA of a language $L$, this semigroup is isomorphic to the transition semigroup of the quotient DFA of $L$, and hence also to the syntactic semigroup of $L$.

\section{Aperiodic Transformations}
\label{sec:aper}

A transformation is \emph{aperiodic} if it contains no cycles of length greater than 1.
A semigroup $T$ of transformations is \emph{aperiodic} if and only if it contains only aperiodic transformations. Thus a language $L$ with quotient DFA $\cA$ is star-free if and only if every transformation in $T_\cA$ is aperiodic.

Each aperiodic transformation can be characterized by a forest of labeled rooted trees as follows.
Consider, for example, the forest of Fig.~\ref{fig:forest}~(a), where the roots are at the bottom. 
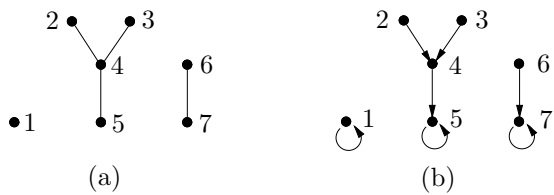
\begin{figure}[b]
\begin{center}
\setlength{\unitlength}{0.00050000in}
\begingroup\makeatletter\ifx\SetFigFont\undefined%
\gdef\SetFigFont#1#2#3#4#5{%
  \reset@font\fontsize{#1}{#2pt}%
  \fontfamily{#3}\fontseries{#4}\fontshape{#5}%
  \selectfont}%
\fi\endgroup%
{\renewcommand{\dashlinestretch}{30}
\begin{picture}(5670,1957)(0,-10)
\put(818,97){\makebox(0,0)[lb]{\smash{{\SetFigFont{10}{12.0}{\familydefault}{\mddefault}{\updefault}(a)}}}}
\put(956,750){\blacken\ellipse{96}{96}}
\put(956,750){\ellipse{96}{96}}
\put(656,1800){\blacken\ellipse{96}{96}}
\put(656,1800){\ellipse{96}{96}}
\put(56,750){\blacken\ellipse{96}{96}}
\put(56,750){\ellipse{96}{96}}
\put(1856,750){\blacken\ellipse{96}{96}}
\put(1856,750){\ellipse{96}{96}}
\put(1856,1350){\blacken\ellipse{96}{96}}
\put(1856,1350){\ellipse{96}{96}}
\put(1260,1800){\blacken\ellipse{96}{96}}
\put(1260,1800){\ellipse{96}{96}}
\path(956,1298)(956,810)(956,765)(971,772)
\path(1856,1290)(1856,802)(1856,757)(1871,764)
\path(1247,1753)(993,1385)
\path(689,1750)(937,1367)
\put(1069,652){\makebox(0,0)[lb]{\smash{{\SetFigFont{10}{12.0}{\familydefault}{\mddefault}{\updefault}5}}}}
\put(1060,1252){\makebox(0,0)[lb]{\smash{{\SetFigFont{10}{12.0}{\familydefault}{\mddefault}{\updefault}4}}}}
\put(1991,1260){\makebox(0,0)[lb]{\smash{{\SetFigFont{10}{12.0}{\familydefault}{\mddefault}{\updefault}6}}}}
\put(154,660){\makebox(0,0)[lb]{\smash{{\SetFigFont{10}{12.0}{\familydefault}{\mddefault}{\updefault}1}}}}
\put(1976,645){\makebox(0,0)[lb]{\smash{{\SetFigFont{10}{12.0}{\familydefault}{\mddefault}{\updefault}7}}}}
\put(386,1718){\makebox(0,0)[lb]{\smash{{\SetFigFont{10}{12.0}{\familydefault}{\mddefault}{\updefault}2}}}}
\put(1385,1732){\makebox(0,0)[lb]{\smash{{\SetFigFont{10}{12.0}{\familydefault}{\mddefault}{\updefault}3}}}}
\put(4426.619,602.509){\arc{263.289}{5.1804}{10.2363}}
\blacken\path(4580.229,639.870)(4486.000,720.000)(4531.435,604.954)(4580.229,639.870)
\put(5333.619,602.509){\arc{263.289}{5.1804}{10.2363}}
\blacken\path(5487.229,639.870)(5393.000,720.000)(5438.435,604.954)(5487.229,639.870)
\put(3534.619,587.509){\arc{263.289}{5.1804}{10.2363}}
\blacken\path(3688.229,624.870)(3594.000,705.000)(3639.435,589.954)(3688.229,624.870)
\put(4406,1358){\blacken\ellipse{96}{96}}
\put(4406,1358){\ellipse{96}{96}}
\put(4406,758){\blacken\ellipse{96}{96}}
\put(4406,758){\ellipse{96}{96}}
\put(4106,1808){\blacken\ellipse{96}{96}}
\put(4106,1808){\ellipse{96}{96}}
\put(3506,758){\blacken\ellipse{96}{96}}
\put(3506,758){\ellipse{96}{96}}
\put(5306,758){\blacken\ellipse{96}{96}}
\put(5306,758){\ellipse{96}{96}}
\put(5306,1358){\blacken\ellipse{96}{96}}
\put(5306,1358){\ellipse{96}{96}}
\put(4710,1808){\blacken\ellipse{96}{96}}
\put(4710,1808){\ellipse{96}{96}}
\path(4697,1761)(4443,1393)
\blacken\path(4486.476,1508.801)(4443.000,1393.000)(4535.855,1474.718)(4486.476,1508.801)
\path(4139,1758)(4387,1375)
\blacken\path(4296.595,1459.422)(4387.000,1375.000)(4346.959,1492.033)(4296.595,1459.422)
\path(4403,1298)(4403,803)
\blacken\path(4373.000,923.000)(4403.000,803.000)(4433.000,923.000)(4373.000,923.000)
\path(5310,1298)(5310,803)
\blacken\path(5280.000,923.000)(5310.000,803.000)(5340.000,923.000)(5280.000,923.000)
\put(3672,698){\makebox(0,0)[lb]{\smash{{\SetFigFont{10}{12.0}{\familydefault}{\mddefault}{\updefault}1}}}}
\put(4586,705){\makebox(0,0)[lb]{\smash{{\SetFigFont{10}{12.0}{\familydefault}{\mddefault}{\updefault}5}}}}
\put(5517,690){\makebox(0,0)[lb]{\smash{{\SetFigFont{10}{12.0}{\familydefault}{\mddefault}{\updefault}7}}}}
\put(5493,1268){\makebox(0,0)[lb]{\smash{{\SetFigFont{10}{12.0}{\familydefault}{\mddefault}{\updefault}6}}}}
\put(4578,1253){\makebox(0,0)[lb]{\smash{{\SetFigFont{10}{12.0}{\familydefault}{\mddefault}{\updefault}4}}}}
\put(4850,1725){\makebox(0,0)[lb]{\smash{{\SetFigFont{10}{12.0}{\familydefault}{\mddefault}{\updefault}3}}}}
\put(3814,1726){\makebox(0,0)[lb]{\smash{{\SetFigFont{10}{12.0}{\familydefault}{\mddefault}{\updefault}2}}}}
\put(4283,91){\makebox(0,0)[lb]{\smash{{\SetFigFont{10}{12.0}{\familydefault}{\mddefault}{\updefault}(b)}}}}
\put(956,1350){\blacken\ellipse{96}{96}}
\put(956,1350){\ellipse{96}{96}}
\end{picture}
}
\end{center}
\caption{Forests and transformations.} 
\label{fig:forest}
\end{figure}
Convert this forest into a directed graph by adding a direction from each parent to its child and a self-loop to each root, as shown in Fig.~\ref{fig:forest}~(b). This directed graph defines the transformation
\begin{equation*}\label{eq:forest}
t=\left( \begin{array}{ccccccc}
1 & 2 &  3 &  4 & 5 & 6 & 7 \\
1 & 4 &  4 & 5  & 5 & 7 & 7
\end{array} \right )
\end{equation*}
and each such transformation is aperiodic since the directed graph has no cycles of length greater than one. 
Thus there is a one-to-one relation between aperiodic transformations of a set of $n$ elements and forests with $n$ nodes. 

\begin{proposition}
\label{prop:caley}
For $n\ge 1$, the number of aperiodic transformations of a set of $n$ elements is $(n+1)^{n-1}$.
\end{proposition}

\begin{proof}
By Caley's theorem~\cite{Cay89,Sho95}, there are $(n+1)^{n-1}$ labeled unrooted trees with $n+1$ nodes. If we fix one node, say node $n+1$, in each of these trees to be the root, then we have $(n+1)^{n-1}$ labeled trees rooted at $n+1$. Let $T$ be any one of these trees, and let $v_1,\ldots,v_m$ be the parents of $n+1$ in $T$. By removing the root $n+1$ from each such rooted tree, we get a labeled forest $F$ with $n$ nodes formed by $m$ rooted trees, where $v_1,\ldots,v_m$ are the roots. The forest $F$ is unique since $T$ is a unique tree rooted at $n+1$. Then we get a unique aperiodic transformation of $\{1,\ldots,n\}$ by adding self-loops on $v_1,\ldots,v_m$. 

All labeled directed forests with $n$ nodes can be obtained uniquely from some rooted tree with $n+1$ nodes by deleting the root. Hence there are $(n+1)^{n-1}$ labeled forests with $n$ nodes, and that many aperiodic transformations of a set of $n$ elements.
\qed
\end{proof}

Since all the inputs of an aperiodic DFA must perform aperiodic transformations, we have the following result: 
\begin{corollary}
\label{cor:caley}
For $n\ge 1$, the syntactic complexity $\sig(L)$ of a star-free language $L$ with $n$ quotients satisfies $\sig(L)\le (n+1)^{n-1}$.
\end{corollary}

The bound of Corollary~\ref{cor:caley} is our first upper bound on the syntactic complexity of a star free language with $n$ quotients, but this bound is not tight in general because 
the set $A_n$ of all aperiodic transformations \emph{is not a semigroup} under composition for $n\ge 3$. For example, if $a=[1,3,1]$ and $b=[2,2,1]$, then $ab=[2,1,2]$, which contains the permutation $(1,2)$. 
Since the set of all transformations of a DFA by non-empty words \emph{is a semigroup,} our task is to find the size of the largest semigroup contained in $A_n$. 

First, let us consider small values of $n$:
\be
\item
If $n=1$, the only two languages, $\emp$ and $\Sig^*$,  are both star-free, since $\Sig^*=\ol{\emp}$.
Here $\sig(L)=1$,  the bound $2^0=1$ of Corollary~\ref{cor:caley} holds and is tight.
\medskip
\item
If $n=2$, there are 4 transformations possible. Since $[2,1]$ is a permutation, the automata of star-free languages can have at most 3 transformations, and the bound $3^1=3$ of Corollary~\ref{cor:caley} holds. We can assume that state 1 is initial, and state 2 is final, for if the initial and final states are interchanged, we still get the same transition semigroup, and hence the same syntactic semigroup. Suppose state 2 is reached from state 1 by input $a$, that is, $1a=2$; then we must have $2a=a$,
 since $a$ must be aperiodic.
\be
\item
If $\Sig=\{a\}$, then the language $L$ is $\ol{\eps}=aa^*$, and $\sig(L)=1$.
\item
If $\Sig=\{a,b\}$, there are three possibilities for $b$. First, if $1b=2b=2$, then this is the same as the case of a one-letter alphabet.
Second, if $1b=2b=1$, then the language is  $\Sig^*a=\ol{\emp}a$,
and $\sig(L)=2$. 
Third, if $1b=1$ and $2b=2$, then $L=\Sig^*a\Sig^*$, and $\sig(L)=2$.\
\item
If $\Sig=\{a,b,c\}$,  the language
$L=(b+c+a(a+c)^*b)^*a(a+c)^*=\Sig^*a(a+c)^*=\Sig^*a\ol{\Sig^*b\Sig^*}$
has $\sig(L)=3$. 
\ee
\ee
 
In summary, for $n=1$ and $2$, the bound of Corollary~\ref{cor:caley} is tight for $|\Sig|=1$ and $|\Sig|=3$, respectively. 
As we shall see later, this is not the case for $n\ge3$.

\section{Monotonicity in Transformations, Automata and Languages}
\label{sec:mono}

In this section we study the syntactic semigroups of languages accepted by monotonic and related automata. 

We denote by $C^n_k$ the binomial coefficient ``$n$ choose $k$''.

\subsection{Monotonic Transformations, DFA's and Languages}\label{subsec:mono}

We show in Section~\ref{sec:sc3} that a tight upper bound for $n=3$ is 10, and that this bound is met by a monotonic language (defined below).
This provides one reason to study monotonic automata and languages.
A second reason is the fact that all the tight upper bounds on the quotient complexity of operations on star-free languages are met by monotonic languages~\cite{BrLiu11}.

A full transformation $t$ of $Q$ is \emph{monotonic} if there exists a total order $\le$ on $Q$ such that, for all $p,q \in Q$, $p \le q$ implies $pt \le qt$. 
From now on we assume that the total order is the usual order on integers, and that $p<q$ means that $p\le q$ and $p\neq q$, as usual.

Let $M_Q$ be the set of all monotonic full transformations of $Q$. 
The following result~\cite{GoHo92,How71} summarizes some properties of the set of all monotonic transformations. It is restated slightly for our purposes, since the work in~\cite{GoHo92} does not consider the identity transformation to be monotonic.

\begin{theorem}[Gomes and Howie]
\label{thm:monotonic_gen}
When $n \ge 1$, the set $M_Q$ of all monotonic full transformations of a set $Q$ of $n$ elements is an aperiodic semigroup of cardinality
$$|M_Q|=f(n)=\sum_{k=1}^{n}{C^{n-1}_{k-1}}{C^n_k}= 
{C^{2n-1}_n},$$
 and it is generated by the set $H = \{a,b_1,\ldots,b_{n-1},c\}$ of $n+1$ transformations of $Q$, where
\be
\item $1a = 1$, $ia = i-1$ for $2 \le i \le n$; 
\item For $1\le i\le n-1$, $i b_i = i+1$, and $j b_i = j$ for all $j \neq i$; 
\item $c$ is the identity transformation. 
\ee
Moreover, the cardinality of the generating set cannot be reduced.
\end{theorem}

\begin{example}
For $n=1$ there is only one transformation $a=c=[1]$ and it is monotonic.
For $n=2$, the three generators are $a=[1,1]$, $b_1=[2,2]$ and $c=[1,2]$, and $M_Q$ consists of these three transformations.
For $n=3$, the four generators $a=[1,1,2]$, $b_1=[2,2,3]$, $b_2=[1,3,3]$, and $c=[1,2,3]$ generate all ten monotonic transformations, as shown in Table~\ref{tab:sf3}.

\begin{table}[ht]
\caption{The monotonic transformations for $n=3$.}
\label{tab:sf3}
\begin{center}
$
\begin{array}{| c|| c| c| c| c|| c|c|c|c|c|c|c|c|}    
\hline
\ \  \ \ &\ \ a \ \ &\ \ b_1 \ \ &\ \ b_2 \ \ & \ \ c \ \
&  \ aa \ &\ ab_1 \ &\ ab_2 \ &\ b_2a \ 
&\  b_2b_1 \  &   ab_1b_2 \\
\hline  \hline
  1
& 1 & 2  & 1   &1 	& 1 & 2   & 1 & 1   & 2 & 3 \\
\hline
2 
& 1 & 2  & 3   &2	& 1 & 2   & 1 & 2 	& 3 & 3\\
\hline
3
& 2 & 3  & 3   &3	& 1 & 2   & 3 & 2  	& 3 & 3\\
\hline
\hline
\end{array}
$
\end{center}
\end{table}
%

\end{example}
\smallskip

Now we turn to DFA's whose inputs perform monotonic transformations.
A~DFA is \emph{monotonic}~\cite{AV05} if all transformations in its transition semigroup are monotonic with respect to some fixed total order. 
Every monotonic DFA is aperiodic because
a monotonic \tr{} has no non-trivial permutations, and the composition of monotonic \tr{}s is monotonic. 
A regular language is \emph{monotonic} if its quotient DFA is monotonic.

Let us now define a DFA having as inputs the generators of $M_Q$:
 
\begin{definition}
\label{def:mon}
For $n \ge 1$, let $\cA_n = (Q,\Sig,\delta,1,\{1\})$ be the DFA in which
$Q = \{1,\ldots,n\}$, $\Sig=\{a,b_1,\ldots,b_{n-1},c\}$, and each letter in $\Sig$ performs the transformation defined in Theorem~\ref{thm:monotonic_gen}. 
\end{definition}

DFA $\cA_n$ is minimal, since state 1 is the only accepting state, and for $2\le i\le n$ only state $i$ accepts $a^{i-1}$.
From Theorem~\ref{thm:monotonic_gen} we have
\begin{corollary}
\label{cor:A}
For $n\ge 1$, the syntactic complexity $\sig(L)$ of any monotonic language $L$ with $n$ quotients satisfies
$\sig(L)\le {C^{2n-1}_n}.$
Moreover, this bound is met by the language $L(\cA_n)$ of Definition~\ref{def:mon}, and cannot be met by any monotonic language over an alphabet having fewer than $n+1$ letters.
\end{corollary}

\subsection{Monotonic  Partial Transformations and  IDFA's }
\label{subsec:parmono}

As we shall see, for $n\ge 4$ the maximal syntactic complexity cannot be reached by monotonic languages; hence we continue our search for larger semigroups of aperiodic transformations.

A partial transformation $t$ of $Q$ is \emph{monotonic} if there exists a total order $\le$ on $Q$ such that, for all $p,q \in \tdom~t$, $p \le q$ implies $pt \le qt$. As before, we assume that the total order on $Q$ is the usual order on integers. Let $PM_Q$ be the set of all monotonic partial transformations on $Q$ with respect to such an order. 
Gomes and Howie~\cite{GoHo92} showed the following result (again restated slightly):

\begin{theorem}[Gomes and Howie]\label{thm:po}
When $n \ge 1$, the set $PM_Q$ of all monotonic partial transformations of a set $Q$ of $n$ elements is an aperiodic semigroup of cardinality 
$$|PM_Q| = g(n)=\sum_{k=0}^n C^n_k C^{n+k-1}_k,$$
and it is generated by the set 
$I= \{a,b_1,\ldots,b_{n-1},c_1,\ldots,c_{n-1},d\}$ of $2n$ partial transformations of $Q$, where
\be 
\item 
 $ja = j-1$ for $j = 2,\ldots,n$; 
\item For $1\le i\le n-1$, $ib_i = i+1$ and $jb_i = j$ for $j = 1,\ldots,i-1,i+2,\ldots,n$; 
\item For $1\le i\le n-1$, $ic_i = i+1$, and $jc_i = j$ for all $j \neq i$; 
\item $d$ is the identity transformation. 
\ee
Moreover, the cardinality of the generating set cannot be reduced.
\end{theorem}

\begin{example}
For $n=1$, the two monotonic partial transformations are $a=[\Box]$, $d=[1]$.
For $n=2$, the eight monotonic partial transformations are generated by
$a=[\Box,1]$, $b_1=[2,\Box]$, $c_1=[2,2]$, and $d=[1,2]$.
For $n=3$, the 38 monotonic partial transformations are generated by
$a=[\Box,1,2]$, $b_1=[2,\Box,3]$, $b_2=[1,3,\Box]$, $c_1=[2,2,3]$, $c_2=[1,3,3]$ and $d=[1,2,3]$.

Partial transformations correspond to IDFA's. For example, the 
transformations $a=[\Box,1]$, $b=[2,\Box]$ and $c=[2,2]$ correspond to the transitions of the IDFA of Fig.~\ref{fig:pmnotm}~(a).
\end{example}

\begin{figure}[b]
\begin{center}
\setlength{\unitlength}{0.00043745in}
\begingroup\makeatletter\ifx\SetFigFont\undefined%
\gdef\SetFigFont#1#2#3#4#5{%
  \reset@font\fontsize{#1}{#2pt}%
  \fontfamily{#3}\fontseries{#4}\fontshape{#5}%
  \selectfont}%
\fi\endgroup%
{\renewcommand{\dashlinestretch}{30}
\begin{picture}(7764,2054)(0,-10)
\put(916,91){\makebox(0,0)[lb]{\smash{{\SetFigFont{8}{9.6}{\familydefault}{\mddefault}{\updefault}(a)}}}}
\put(1099.500,473.500){\arc{1785.567}{4.1547}{5.2701}}
\blacken\path(718.918,1313.772)(627.000,1231.000)(747.055,1260.778)(718.918,1313.772)
\put(6113.500,2303.500){\arc{2565.395}{1.1935}{1.9480}}
\blacken\path(6482.251,1043.648)(6586.000,1111.000)(6462.762,1100.395)(6482.251,1043.648)
\put(6113.500,443.500){\arc{1785.567}{4.1547}{5.2701}}
\blacken\path(5732.918,1283.772)(5641.000,1201.000)(5761.055,1230.778)(5732.918,1283.772)
\put(3743.000,684.000){\arc{450.000}{5.3559}{10.3521}}
\blacken\path(3559.021,750.417)(3608.000,864.000)(3511.333,786.828)(3559.021,750.417)
\put(2405.500,1156.000){\arc{425.000}{3.5789}{8.9874}}
\blacken\path(2315.181,996.293)(2213.000,1066.000)(2270.356,956.408)(2315.181,996.293)
\put(7415.500,1133.000){\arc{425.000}{3.5789}{8.9874}}
\blacken\path(7325.181,973.293)(7223.000,1043.000)(7280.356,933.408)(7325.181,973.293)
\put(310,1173){\ellipse{604}{604}}
\put(316,1168){\ellipse{512}{512}}
\put(1885,1173){\ellipse{604}{604}}
\put(5324,1143){\ellipse{604}{604}}
\put(3751,1156){\ellipse{604}{604}}
\put(6899,1143){\ellipse{604}{604}}
\put(5330,1138){\ellipse{512}{512}}
\blacken\path(4186.000,1186.000)(4066.000,1156.000)(4186.000,1126.000)(4186.000,1186.000)
\path(4066,1156)(5011,1156)
\path(5318,556)(5318,826)
\blacken\path(5348.000,706.000)(5318.000,826.000)(5288.000,706.000)(5348.000,706.000)
\path(301,571)(301,841)
\blacken\path(331.000,721.000)(301.000,841.000)(271.000,721.000)(331.000,721.000)
\path(6631,1328)(6630,1329)(6628,1330)
	(6624,1332)(6619,1335)(6610,1340)
	(6599,1347)(6585,1355)(6568,1364)
	(6548,1376)(6525,1389)(6499,1403)
	(6470,1418)(6439,1435)(6405,1452)
	(6370,1470)(6333,1489)(6294,1507)
	(6254,1526)(6212,1545)(6169,1563)
	(6124,1581)(6079,1599)(6032,1616)
	(5983,1632)(5932,1648)(5880,1663)
	(5825,1677)(5769,1690)(5709,1702)
	(5647,1712)(5583,1722)(5516,1730)
	(5446,1735)(5375,1739)(5303,1741)
	(5231,1740)(5160,1737)(5092,1732)
	(5025,1725)(4962,1716)(4902,1706)
	(4844,1695)(4789,1683)(4736,1669)
	(4686,1655)(4637,1640)(4591,1624)
	(4546,1607)(4502,1590)(4460,1573)
	(4420,1555)(4380,1537)(4342,1519)
	(4306,1500)(4271,1482)(4238,1465)
	(4207,1448)(4178,1432)(4151,1416)
	(4127,1402)(4106,1390)(4087,1379)
	(4071,1369)(4058,1361)(4048,1355)
	(4040,1350)(4029,1343)
\blacken\path(4114.133,1432.735)(4029.000,1343.000)(4146.346,1382.115)(4114.133,1432.735)
\put(312,1118){\makebox(0,0)[b]{\smash{{\SetFigFont{8}{9.6}{\familydefault}{\mddefault}{\updefault}$1$}}}}
\put(1887,1118){\makebox(0,0)[b]{\smash{{\SetFigFont{8}{9.6}{\familydefault}{\mddefault}{\updefault}$2$}}}}
\put(1129,1448){\makebox(0,0)[b]{\smash{{\SetFigFont{8}{9.6}{\familydefault}{\mddefault}{\updefault}$a$}}}}
\put(1085,841){\makebox(0,0)[b]{\smash{{\SetFigFont{8}{9.6}{\familydefault}{\mddefault}{\updefault}$b,c$}}}}
\put(2697,1118){\makebox(0,0)[b]{\smash{{\SetFigFont{8}{9.6}{\familydefault}{\mddefault}{\updefault}$c$}}}}
\put(3751,1088){\makebox(0,0)[b]{\smash{{\SetFigFont{8}{9.6}{\familydefault}{\mddefault}{\updefault}$3$}}}}
\put(5326,1088){\makebox(0,0)[b]{\smash{{\SetFigFont{8}{9.6}{\familydefault}{\mddefault}{\updefault}$1$}}}}
\put(6901,1088){\makebox(0,0)[b]{\smash{{\SetFigFont{8}{9.6}{\familydefault}{\mddefault}{\updefault}$2$}}}}
\put(3714,211){\makebox(0,0)[b]{\smash{{\SetFigFont{8}{9.6}{\familydefault}{\mddefault}{\updefault}$a,b,c$}}}}
\put(4523,1238){\makebox(0,0)[b]{\smash{{\SetFigFont{8}{9.6}{\familydefault}{\mddefault}{\updefault}$a$}}}}
\put(6137,788){\makebox(0,0)[b]{\smash{{\SetFigFont{8}{9.6}{\familydefault}{\mddefault}{\updefault}$b,c$}}}}
\put(7749,1073){\makebox(0,0)[b]{\smash{{\SetFigFont{8}{9.6}{\familydefault}{\mddefault}{\updefault}$c$}}}}
\put(5326,1832){\makebox(0,0)[b]{\smash{{\SetFigFont{8}{9.6}{\familydefault}{\mddefault}{\updefault}$b$}}}}
\put(5836,1380){\makebox(0,0)[b]{\smash{{\SetFigFont{8}{9.6}{\familydefault}{\mddefault}{\updefault}$a$}}}}
\put(5161,106){\makebox(0,0)[lb]{\smash{{\SetFigFont{8}{9.6}{\familydefault}{\mddefault}{\updefault}(b)}}}}
\put(1099.500,2333.500){\arc{2565.395}{1.1935}{1.9480}}
\blacken\path(1468.251,1073.648)(1572.000,1141.000)(1448.762,1130.395)(1468.251,1073.648)
\end{picture}
}
\end{center}
\caption{Partially monotonic DFA's: (a) incomplete; (b) complete.} 
\label{fig:pmnotm}
\end{figure}
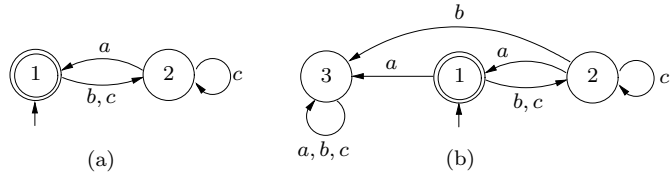

An IDFA is \emph{monotonic} if all partial transformations in its transition semigroup are monotonic with respect to some fixed total order. 
A regular language is \emph{partially monotonic} if its quotient IDFA is monotonic.

\begin{definition}
\label{def:partial}
For $n \ge 1$, let $\cB_n = (Q,\Sig,\delta,1,\{1\})$ be the IDFA in which $Q = \{1,\ldots,n\}$,
$\Sig= \{a,b_1,\ldots,b_{n-1},c_1,\ldots,c_{n-1},d\}$, and each letter in $\Sig$ performs the partial transformation defined in Theorem~\ref{thm:po}.
\end{definition}

\begin{corollary}
\label{cor:B}
For $n\ge 1$, the syntactic complexity $\sig(L)$ of any partially monotonic regular language $L$ with $n$ non-empty quotients has an upper bound $\sig(L) \le \sum_{k=0}^{n}C^{n}_kC^{n+k-1}_k$.
Moreover, this bound is met by the language $L(\cB_n)$ of Definition~\ref{def:partial}, and cannot be met by any partially monotonic language over an alphabet having fewer than $2n$ letters.
\end{corollary}

We can also reformulate this result in terms of quotient DFA's.
In the following, we let $Q'=\{1,\ldots,n-1\}$, and $Q=Q'\cup \{n\}$.
A quotient DFA  $\cA=(Q,\Sig,\delta,q_1,F)$ is \emph{partially monotonic} if the corresponding quotient IDFA is monotonic. So a regular language is partially monotonic if its quotient DFA is partially monotonic.

\begin{example}

If we complete the transformations in Fig.~\ref{fig:pmnotm}~(a) by replacing the undefined entry $\Box$ by a new empty (or ``sink'') state 3, as usual, we obtain the DFA of Fig.~\ref{fig:pmnotm}~(b).
That DFA  is not monotonic, because $1<2$ implies $2< 3$ under input $b$, and $3<1$ under $a$, which in turn implies that $3<2$ under $b$.
A contradiction is also obtained if we assume that $2<1$.
However, this DFA is partially monotonic, since removal of the empty state 3 results in the IDFA of Fig.~\ref{fig:pmnotm}~(a), whose partial transformations are monotonic.

The DFA of Fig.~\ref{fig:pm}  has an empty state, and it is monotonic for the order shown. 
The IDFA consisting of states 2, 3, and 4 is also monotonic. 

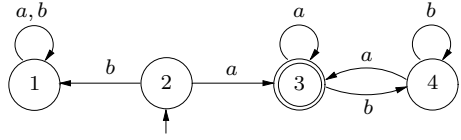
\begin{figure}[h]
\begin{center}
\setlength{\unitlength}{0.00043745in}
\begingroup\makeatletter\ifx\SetFigFont\undefined%
\gdef\SetFigFont#1#2#3#4#5{%
  \reset@font\fontsize{#1}{#2pt}%
  \fontfamily{#3}\fontseries{#4}\fontshape{#5}%
  \selectfont}%
\fi\endgroup%
{\renewcommand{\dashlinestretch}{30}
\begin{picture}(5331,1625)(0,-10)
\put(277,1391){\makebox(0,0)[b]{\smash{{\SetFigFont{8}{9.6}{\familydefault}{\mddefault}{\updefault}$a,b$}}}}
\put(4235.500,-108.500){\arc{1785.567}{4.1547}{5.2701}}
\blacken\path(3854.918,731.772)(3763.000,649.000)(3883.055,678.778)(3854.918,731.772)
\put(3441.000,1069.000){\arc{450.000}{2.2143}{7.2105}}
\blacken\path(3624.979,1002.583)(3576.000,889.000)(3672.667,966.172)(3624.979,1002.583)
\put(5008.000,1069.000){\arc{450.000}{2.2143}{7.2105}}
\blacken\path(5191.979,1002.583)(5143.000,889.000)(5239.667,966.172)(5191.979,1002.583)
\put(306.000,1061.000){\arc{450.000}{2.2143}{7.2105}}
\blacken\path(489.979,994.583)(441.000,881.000)(537.667,958.172)(489.979,994.583)
\put(1873,604){\ellipse{604}{604}}
\put(5021,591){\ellipse{604}{604}}
\put(3446,591){\ellipse{604}{604}}
\put(3452,586){\ellipse{512}{512}}
\put(310,587){\ellipse{604}{604}}
\path(2188,604)(3133,604)
\blacken\path(3013.000,574.000)(3133.000,604.000)(3013.000,634.000)(3013.000,574.000)
\path(1866,12)(1866,282)
\blacken\path(1896.000,162.000)(1866.000,282.000)(1836.000,162.000)(1896.000,162.000)
\path(1559,603)(614,603)
\blacken\path(734.000,633.000)(614.000,603.000)(734.000,573.000)(734.000,633.000)
\put(1873,536){\makebox(0,0)[b]{\smash{{\SetFigFont{8}{9.6}{\familydefault}{\mddefault}{\updefault}$2$}}}}
\put(5023,536){\makebox(0,0)[b]{\smash{{\SetFigFont{8}{9.6}{\familydefault}{\mddefault}{\updefault}$4$}}}}
\put(4259,236){\makebox(0,0)[b]{\smash{{\SetFigFont{8}{9.6}{\familydefault}{\mddefault}{\updefault}$b$}}}}
\put(4242,873){\makebox(0,0)[b]{\smash{{\SetFigFont{8}{9.6}{\familydefault}{\mddefault}{\updefault}$a$}}}}
\put(5001,1391){\makebox(0,0)[b]{\smash{{\SetFigFont{8}{9.6}{\familydefault}{\mddefault}{\updefault}$b$}}}}
\put(3440,1383){\makebox(0,0)[b]{\smash{{\SetFigFont{8}{9.6}{\familydefault}{\mddefault}{\updefault}$a$}}}}
\put(1207,708){\makebox(0,0)[b]{\smash{{\SetFigFont{8}{9.6}{\familydefault}{\mddefault}{\updefault}$b$}}}}
\put(2645,701){\makebox(0,0)[b]{\smash{{\SetFigFont{8}{9.6}{\familydefault}{\mddefault}{\updefault}$a$}}}}
\put(312,522){\makebox(0,0)[b]{\smash{{\SetFigFont{8}{9.6}{\familydefault}{\mddefault}{\updefault}$1$}}}}
\put(3448,521){\makebox(0,0)[b]{\smash{{\SetFigFont{8}{9.6}{\familydefault}{\mddefault}{\updefault}$3$}}}}
\put(4235.500,1751.500){\arc{2565.395}{1.1935}{1.9480}}
\blacken\path(4604.251,491.648)(4708.000,559.000)(4584.762,548.395)(4604.251,491.648)
\end{picture}
}
\end{center}
\caption{Partially monotonic DFA  that is  monotonic and has an empty state.} 
\label{fig:pm}
\end{figure}
The transition semigroup of the DFA of Fig.~\ref{fig:pm} consists of the transformations
$a=[1,3,3,3]$, $b=[1,1,4,4]$, $ab=[1,4,4,4]$, and $ba=[1,1,3,3]$.
The transition semigroup of the corresponding IDFA with state set $\{2,3,4\}$ consists of 
$a=[3,3,3]$, $b=[\Box,4,4]$, $ab=[4,4,4]$, and $ba=[\Box,3,3]$, and the two semigroups are isomorphic.
\end{example}

Now consider the following semigroup $CM_Q$ of \emph{monotonic completed transformations} on a set $Q=\{1,\ldots,n\}$. Start with the semigroup $PM_{Q'}$ of all monotonic partial transformations on  
$Q'=\{1,\ldots,n-1\}$. Convert the partial transformations to full by adding $n$, replacing all the $\Box$'s by $n$'s, and letting the image of $n$ be $n$.
As we have noted before, there is one-to-one correspondence between the monotonic partial transformations  of $Q'$ and the corresponding full transformations of $Q$.
Thus semigroups $CM_Q$ and $PM_{Q'}$ are isomorphic.

\begin{example}
For $n=2$, the two monotonic completed transformations are $a=[2,2]$, $d=[1,2]$.
For $n=3$, the eight monotonic completed transformations are generated by
$a=[3,1,3]$, $b_1=[2,3,3]$, $c_1=[2,2,3]$, and $d=[1,2,3]$.
\end{example}

\begin{definition}
\label{def:partmon}
For $n \ge 1$, let $\cB'_n = (Q,\Sig,\delta,1,\{1\})$ be the DFA in which $Q = \{1,\ldots,n\}$,
$\Sig= \{a,b_1,\ldots,b_{n-2},c_1,\ldots,c_{n-2},d\}$, and each letter in $\Sig$ defines a transformation as follows: 
\be 
\item $1a = na = n$, and $ja = j-1$ for $j = 2,\ldots,n-1$; 
\item For $1\le i\le n-2$, $ib_i = i+1$, $(i+1)b_i = n$, and $jb_i = j$ for $j = 1,\ldots,i-1,i+2,\ldots,n$; 
\item For $1\le i\le n-2$, $ic_i = i+1$, and $jc_i = j$ for all $j \neq i$; 
\item $d$ is the identity transformation. 
\ee
\end{definition}

Theorem~\ref{thm:po} then implies the following result:
\begin{corollary}
\label{cor:B'}
For $n\ge 1$, the syntactic complexity $\sig(L)$ of any partially monotonic language $L$ with $n$ quotients satisfies $\sigma(L) \le g(n-1)=\sum_{k=0}^{n-1}C^{n-1}_kC^{n+k-2}_k$.
Moreover, this bound is met by the language $L(\cB'_n)$ of Definition~\ref{def:partmon}, and cannot be met by any partially monotonic language over an alphabet having fewer than $2n-2$ letters.
\end{corollary}

\subsection{Nearly Monotonic Transformations and DFA's}
\label{subsec:nearlymono}

Partially monotonic languages still don't reach the maximal bound for syntactic complexity, so we continue  the search for larger semigroups of aperiodic transformations.

Let $NM_Q$ be the union of $CM_Q$ and the set of all constant transformations of $Q$;
 we shall call these transformations \emph{nearly monotonic} with respect to the usual order on integers.
 

\begin{theorem}\label{thm:nm}
When $n \ge 2$, the set $NM_Q$ of all nearly monotonic transformations of a set $Q$ of $n$ elements is an aperiodic semigroup of cardinality 
$$|NM_Q| = h(n)=\sum_{k=0}^{n-1}C^{n-1}_kC^{n+k-2}_k+n-1,$$
and it is generated by the set 
$J= \{a,b_1,\ldots,b_{n-2},c_1,\ldots,c_{n-2},d, e\}$ of $2n-2$ transformations of $Q$, where
\be 
\item $1a = na = n$, and $ja = j-1$ for $j = 2,\ldots,n-1$; 
\item For $1\le i\le n-2$, $ib_i = i+1$, $(i+1)b_i = n$, and $jb_i = j$ for $j = 1,\ldots,i-1,i+2,\ldots,n$; 
\item For $1\le i\le n-2$, $ic_i = i+1$, and $jc_i = j$ for all $j \neq i$; 
\item $d$ is the identity transformation;
\item $e$ is the constant transformation ${Q\choose 1}$.
\ee
Moreover, the cardinality of the generating set cannot be reduced.
\end{theorem}
\begin{proof}
Let $I' = \{a, b_1, \ldots, b_{n-2}, c_1, \ldots, c_{n-2}, d\}$. We know from~\cite{GoHo92} that the semigroup $CM_Q=\langle I' \rangle$ generated by $I'$ is isomorphic to $PM_{Q'}$, where $Q' = \{1,\ldots,n-1\}$. Let $J = I' \cup \{e\}$, and $K = \langle J \rangle \setminus \langle I' \rangle$.

The composition of a constant transformation with any transformation is a constant. Thus any transformation in $NM_Q$ is either constant or is in the semigroup $CM_Q$. Since constant transformations are aperiodic, and $CM_Q$ is  aperiodic,  $NM_Q$ is also aperiodic. 

Note that $K$ contains only constant transformations. Since $\langle I' \rangle \simeq PM_{Q'}$, for all $j \in Q$, $s_j = {Q' \choose j}{n \choose n} \in \langle I' \rangle \subseteq \langle J \rangle$. Thus $es_j = {Q \choose j} \in \langle J \rangle$ for all $1 \le j \le n$. Among these constant transformations, only ${Q \choose n} \in \langle I' \rangle$; so $K$ contains ${Q \choose j}$ for $1 \le j \le n-1$ and $|K| = n-1$. 
Thus $h(n) = |\langle J \rangle| = |PM_{Q'}| + n - 1$. 

Since the cardinality of $I'$ cannot be reduced, and $e\not \in \langle I' \rangle$, also the cardinality of $J$ cannot be reduced. 
\qed
\end{proof}

\begin{example}
For $n=2$, the three nearly monotonic transformations are $a=[2,2]$, $d=[1,2]$ and $e=[1,1]$.
For $n=3$, the ten nearly monotonic transformations are generated by
$a=[3,1,3]$, $b_1=[2,3,3]$, $c_1=[2,2,3]$,  $d=[1,2,3]$, and $e=[1,1,1]$.
\end{example}

Let $\cA = (Q, \Sig, \delta, q_1, F)$ be a DFA. An input $a \in \Sig$ is {\em constant} if $\delta(p, a) = \delta(q, a)$ for all $p, q \in Q$, that is, $a$ performs a constant transformation of $Q$. Then $\cA$ is {\em nearly monotonic} if, after removing constant inputs, the resulting DFA $\cA'$ is partially monotonic. 
A regular language is \emph{nearly monotonic} if its quotient DFA is nearly monotonic.

\begin{definition}\label{def:nm}
For $n \ge 2$, let $\cC_n = (Q,\Sig,\delta,1,\{1\})$ be a DFA, where $Q = \{1,\ldots,n\}$, $\Sig = \{a,b_1,\ldots,b_{n-2},c_1,\ldots,c_{n-2},d,e\}$, and each letter in $\Sig$ performs the transformation defined in Theorem~\ref{thm:nm}.
\end{definition}

Theorem~\ref{thm:nm} now leads us to the following result:
\begin{theorem}
\label{thm:C}
For $n\ge 2$, the syntactic complexity $\sig(L)$ of any nearly monotonic language $L$ with $n$ quotients satisfies 
$\sigma(L) \le h(n)=\sum_{k=0}^{n-1}C^{n-1}_kC^{n+k-2}_k+n-1$.
Moreover, this bound is met by the language $L(\cC_n)$ of Definition~\ref{def:nm}, and cannot be met by any nearly monotonic language over an alphabet having fewer than $2n-1$ letters.
\end{theorem}
\begin{proof}
Note that 1 is the initial state of $\cC_n$. For $2 \le i \le n-1$, state $i$ is reached by  $w_i = b_1\cdots b_{i-1}$. State $n$ is reached by  $w_{n-1}b_{n-2}$. Thus all states are reachable. For $1 \le i \le n-1$, state $i$ accepts the word $a^{i-1}$, while all other states reject it. Also, state $n$ rejects $a^i$ for all $i \ge 0$. So all $n$ states are distinguishable, and $\cC_n$ is minimal. Therefore $L$ has quotient complexity $\kappa(L) = n$. The syntactic semigroup of $L$ is generated by $J$, and so $L$ has syntactic complexity $\sigma(L) = h(n) = \sum_{k=0}^{n-1}C^{n-1}_kC^{n+k-2}_k + n - 1$. 
\qed
\end{proof}

\smallskip

Although we cannot prove that $NM_Q$ is the largest semigroup of aperiodic transformations, we can show that no transformation can be added to $NM_Q$ without destroying aperiodicity.

A set $S = \{T_1, T_2, \ldots, T_k\}$ of transformation semigroups is a {\em chain}  if $T_1 \subset T_2 \subset \cdots \subset T_k$. Semigroup $T_k$ is the largest in $S$, and we denote it by $\max(S) = T_k$. The following result shows that the syntactic semigroup $T_{L(\cC_n)} = T_{\cC_n}$ of $L(\cC_n)$ in Definition~\ref{def:nm} is a local maximum among aperiodic subsemigroups of $\cT_Q$.

\begin{proposition}
\label{prop:pminchain}
Let $S$ be a chain of aperiodic subsemigroups of $\cT_Q$. If $T_{\cC_n} \in S$, then $T_{\cC_n} = \max(S)$.
\end{proposition}

\begin{proof}
Suppose $\max(S) = T_k$ for some aperiodic subsemigroup $T_k$ of $\cT_Q$, and $T_k \neq T_{\cC_n}$. Then there exists $t \in T_k$ such that $t \not\in T_{\cC_n}$; thus there exists $i, j \in Q$ such that $i < j \neq n$ but $it > jt$, and $it, jt \neq n$. Let $\tau \in \cT_Q$ be the transformation of $Q$ such that $(jt)\tau = i$, $(it)\tau = j$, and $h\tau = n$ for all $h \neq i, j$; then $\tau \in T_{\cC_n}$. Let $\lambda \in \cT_Q$ be such that $i\lambda = i$, $j\lambda = j$, and $h\lambda = n$ for all $h \neq i,j$; then also $\lambda \in T_{\cC_n}$. Since $T_k = \max(S)$, $T_{\cC_n} \subset T_k$ and $\tau, \lambda \in T_k$. Then $s = \lambda t \tau$ is also in $T_k$. However, $is = i(\lambda t \tau) = j$, $js = j(\lambda t \tau) = i$, and $hs = n$ for all $h \neq i, j$; then $s = (i, j){P \choose n}$, where $P=Q\setminus\{i,j\}$,  is not aperiodic, a contradiction. Therefore $T_{\cC_n} = \max(S)$. \qed
\end{proof}

\subsection{Containment and Closure Properties}

Let ${\mathbb L}_M$, ${\mathbb L}_{PM}$, ${\mathbb L}_{NM}$, and ${\mathbb L}_{SF}$ be the classes of monotonic, partially monotonic, nearly monotonic and star-free languages.
Then the following holds:
\begin{proposition}
$${\mathbb L}_M \subsetneq {\mathbb L}_{PM} \subsetneq {\mathbb L}_{NM} \subsetneq {\mathbb L}_{SF}.$$
\end{proposition}
\begin{proof}
The DFA of Fig.~\ref{fig:pmnotm}~(b) is partially monotonic but not monotonic.
If we add to that DFA the input $d= [1,1,1]$, then the new DFA is nearly monotonic but not partially monotonic.
The DFA of Figure~\ref{fig:SFnotNM} is aperiodic, as one can verify.
It has no constant input; hence it must be partially monotonic if it is nearly monotonic. 
Since it has no empty state, it must be monotonic if it is partially monotonic.
However, if $1<2$. then $3<2$ by $a$, and also $2<3$ by $b$.
 We get a similar contradiction if we set $2 < 1$. Therefore $\cA$ is not monotonic. 
\qed
\end{proof}

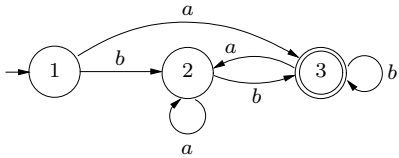
\begin{figure}[b]
\begin{center}
\setlength{\unitlength}{0.00043745in}
\begingroup\makeatletter\ifx\SetFigFont\undefined%
\gdef\SetFigFont#1#2#3#4#5{%
  \reset@font\fontsize{#1}{#2pt}%
  \fontfamily{#3}\fontseries{#4}\fontshape{#5}%
  \selectfont}%
\fi\endgroup%
{\renewcommand{\dashlinestretch}{30}
\begin{picture}(4601,1959)(0,-10)
\put(2156,87){\makebox(0,0)[b]{\smash{{\SetFigFont{8}{9.6}{\familydefault}{\mddefault}{\updefault}$a$}}}}
\put(2950.500,348.500){\arc{1785.567}{4.1547}{5.2701}}
\blacken\path(2569.918,1188.772)(2478.000,1106.000)(2598.055,1135.778)(2569.918,1188.772)
\put(4252.500,1038.000){\arc{425.000}{3.5789}{8.9874}}
\blacken\path(4162.181,878.293)(4060.000,948.000)(4117.356,838.408)(4162.181,878.293)
\put(2163.000,538.000){\arc{425.906}{5.1487}{10.5592}}
\blacken\path(2003.357,628.776)(2073.000,731.000)(1963.444,673.575)(2003.357,628.776)
\put(2161,1048){\ellipse{604}{604}}
\put(588,1061){\ellipse{604}{604}}
\put(3736,1048){\ellipse{604}{604}}
\put(3739,1051){\ellipse{512}{512}}
\path(903,1061)(1848,1061)
\blacken\path(1728.000,1031.000)(1848.000,1061.000)(1728.000,1091.000)(1728.000,1031.000)
\path(12,1053)(282,1053)
\blacken\path(162.000,1023.000)(282.000,1053.000)(162.000,1083.000)(162.000,1023.000)
\blacken\path(3349.230,1267.551)(3468.000,1233.000)(3379.463,1319.378)(3349.230,1267.551)
\path(3468,1233)(3456,1240)(3447,1245)
	(3436,1252)(3422,1260)(3405,1269)
	(3385,1281)(3362,1294)(3336,1308)
	(3307,1323)(3276,1340)(3242,1357)
	(3207,1375)(3170,1394)(3131,1412)
	(3091,1431)(3049,1450)(3006,1468)
	(2961,1486)(2916,1504)(2869,1521)
	(2820,1537)(2769,1553)(2717,1568)
	(2662,1582)(2606,1595)(2546,1607)
	(2484,1617)(2420,1627)(2353,1635)
	(2283,1640)(2212,1644)(2140,1646)
	(2068,1645)(1997,1642)(1929,1637)
	(1862,1630)(1799,1621)(1739,1611)
	(1681,1600)(1626,1588)(1573,1574)
	(1523,1560)(1474,1545)(1428,1529)
	(1383,1512)(1339,1495)(1297,1478)
	(1257,1460)(1217,1442)(1179,1424)
	(1143,1405)(1108,1387)(1075,1370)
	(1044,1353)(1015,1337)(988,1321)
	(964,1307)(943,1295)(924,1284)
	(908,1274)(895,1266)(885,1260)
	(877,1255)(872,1252)(869,1250)
	(867,1249)(866,1248)
\put(588,993){\makebox(0,0)[b]{\smash{{\SetFigFont{8}{9.6}{\familydefault}{\mddefault}{\updefault}$1$}}}}
\put(2163,993){\makebox(0,0)[b]{\smash{{\SetFigFont{8}{9.6}{\familydefault}{\mddefault}{\updefault}$2$}}}}
\put(3738,993){\makebox(0,0)[b]{\smash{{\SetFigFont{8}{9.6}{\familydefault}{\mddefault}{\updefault}$3$}}}}
\put(1360,1143){\makebox(0,0)[b]{\smash{{\SetFigFont{8}{9.6}{\familydefault}{\mddefault}{\updefault}$b$}}}}
\put(2974,693){\makebox(0,0)[b]{\smash{{\SetFigFont{8}{9.6}{\familydefault}{\mddefault}{\updefault}$b$}}}}
\put(4586,978){\makebox(0,0)[b]{\smash{{\SetFigFont{8}{9.6}{\familydefault}{\mddefault}{\updefault}$b$}}}}
\put(2163,1737){\makebox(0,0)[b]{\smash{{\SetFigFont{8}{9.6}{\familydefault}{\mddefault}{\updefault}$a$}}}}
\put(2673,1285){\makebox(0,0)[b]{\smash{{\SetFigFont{8}{9.6}{\familydefault}{\mddefault}{\updefault}$a$}}}}
\put(2950.500,2208.500){\arc{2565.395}{1.1935}{1.9480}}
\blacken\path(3319.251,948.648)(3423.000,1016.000)(3299.762,1005.395)(3319.251,948.648)
\end{picture}
}
\end{center}
\caption{A DFA $\cA$ that is aperiodic but not monotonic.} 
\label{fig:SFnotNM}
\end{figure}

\begin{proposition}
The classes ${\mathbb L}_M$, ${\mathbb L}_{PM}$, and ${\mathbb L}_{NM}$
are closed under complementation and left quotients, but not under union, symmetric difference, intersection, difference, concatenation and star.
\end{proposition}
\begin{proof}
Closure under complementation is obvious, since the properties of monotonicity do not involve final states. Also, the left quotient of any monotonic language is defined by a DFA that is a subautomaton of the original language, and closure under quotients follows.

The language of the DFA of Figure~\ref{fig:SFnotNM} is not nearly monotonic. However, it is the union of languages $L_1=a(\eps\cup \Sig^*b)$ and $L_2=b\Sig^*b$.
The DFA of $L_1$ is shown in Figure~\ref{fig:pm}, and it is monotonic.
One verifies that $L_2$ is also monotonic. 
This proves that each of the classes is not closed under union. 
Since $L_1$ and $L_2$ are disjoint, each  class is not closed under symmetric difference.
Since each class is closed under complement, it can be closed under neither intersection nor difference.

Consider now DFA's: $\cA_3 = (\{1,2,3\}, \{a,b\}, \delta_3, 1, \{2\})$, where $a = [1, 1, 1]$ and $b = [2, 3, 3]$, and $\cA_4 = (\{1,2,3\}, \{a,b\}, \delta_4, 1, \{2\})$, where $a = [1, 1, 3]$ and $b = [2, 3, 3]$. Both $\cA_3$ and $\cA_4$ are monotonic.
 Let $L = L(\cA_3)L(\cA_4)$. Then $L$ is star-free with quotient DFA $\cA = (\{1,2,3,4,5,6\}, \{a, b\}, \delta, 1, \{3, 6\})$ shown in Figure~\ref{fig:concatpm}, where $a = [1, 5, 5, 1, 5, 5]$ and $b = [2, 3, 4, 4, 6, 3]$. 

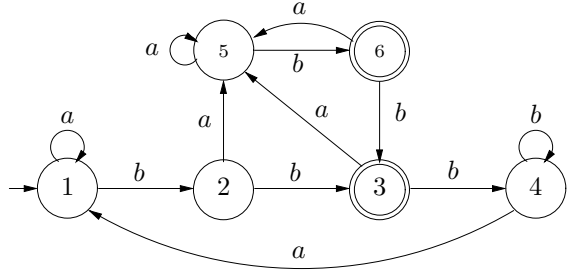
\begin{figure}[hbt]
\begin{center}
\setlength{\unitlength}{0.00052493in}
\begingroup\makeatletter\ifx\SetFigFont\undefined%
\gdef\SetFigFont#1#2#3#4#5{%
  \reset@font\fontsize{#1}{#2pt}%
  \fontfamily{#3}\fontseries{#4}\fontshape{#5}%
  \selectfont}%
\fi\endgroup%
{\renewcommand{\dashlinestretch}{30}
\begin{picture}(5632,2804)(0,-10)
\put(2930,114){\makebox(0,0)[b]{\smash{{\SetFigFont{10}{12.0}{\familydefault}{\mddefault}{\updefault}$a$}}}}
\put(2959.500,4202.654){\arc{8388.394}{1.0360}{2.1056}}
\blacken\path(941.057,560.451)(822.000,594.000)(911.262,508.371)(941.057,560.451)
\put(597.000,1233.643){\arc{340.715}{2.1274}{7.2974}}
\blacken\path(735.781,1202.668)(687.000,1089.000)(783.533,1166.340)(735.781,1202.668)
\put(2950.500,1633.500){\arc{1705.019}{4.1376}{5.3781}}
\blacken\path(2577.698,2433.106)(2487.000,2349.000)(2606.607,2380.530)(2577.698,2433.106)
\put(1782.000,2214.000){\arc{300.000}{0.6435}{5.6397}}
\blacken\path(1782.521,2336.012)(1902.000,2304.000)(1811.642,2388.471)(1782.521,2336.012)
\put(2170,806){\ellipse{604}{604}}
\put(3747,806){\ellipse{512}{512}}
\put(3745,806){\ellipse{604}{604}}
\put(597,819){\ellipse{604}{604}}
\put(5322,819){\ellipse{604}{604}}
\put(2172,2214){\ellipse{604}{604}}
\put(3747,2201){\ellipse{512}{512}}
\put(3745,2201){\ellipse{604}{604}}
\path(12,819)(282,819)
\blacken\path(162.000,789.000)(282.000,819.000)(162.000,849.000)(162.000,789.000)
\path(912,819)(1857,819)
\blacken\path(1737.000,789.000)(1857.000,819.000)(1737.000,849.000)(1737.000,789.000)
\path(2487,819)(3432,819)
\blacken\path(3312.000,789.000)(3432.000,819.000)(3312.000,849.000)(3312.000,789.000)
\path(4062,819)(5007,819)
\blacken\path(4887.000,789.000)(5007.000,819.000)(4887.000,849.000)(4887.000,789.000)
\path(3567,1044)(2397,1989)
\blacken\path(2509.203,1936.938)(2397.000,1989.000)(2471.503,1890.261)(2509.203,1936.938)
\path(2172,1089)(2172,1899)
\blacken\path(2202.000,1779.000)(2172.000,1899.000)(2142.000,1779.000)(2202.000,1779.000)
\blacken\path(3717.000,1209.000)(3747.000,1089.000)(3777.000,1209.000)(3717.000,1209.000)
\path(3747,1089)(3747,1899)
\path(2487,2214)(3432,2214)
\blacken\path(3312.000,2184.000)(3432.000,2214.000)(3312.000,2244.000)(3312.000,2184.000)
\put(2172,751){\makebox(0,0)[b]{\smash{{\SetFigFont{10}{12.0}{\familydefault}{\mddefault}{\updefault}2}}}}
\put(3747,751){\makebox(0,0)[b]{\smash{{\SetFigFont{10}{12.0}{\familydefault}{\mddefault}{\updefault}3}}}}
\put(597,751){\makebox(0,0)[b]{\smash{{\SetFigFont{10}{12.0}{\familydefault}{\mddefault}{\updefault}1}}}}
\put(5322,751){\makebox(0,0)[b]{\smash{{\SetFigFont{10}{12.0}{\familydefault}{\mddefault}{\updefault}4}}}}
\put(3162,1539){\makebox(0,0)[b]{\smash{{\SetFigFont{10}{12.0}{\familydefault}{\mddefault}{\updefault}$a$}}}}
\put(3747,2146){\makebox(0,0)[b]{\smash{{\SetFigFont{7}{8.4}{\familydefault}{\mddefault}{\updefault}6}}}}
\put(2172,2146){\makebox(0,0)[b]{\smash{{\SetFigFont{7}{8.4}{\familydefault}{\mddefault}{\updefault}5}}}}
\put(1452,2169){\makebox(0,0)[b]{\smash{{\SetFigFont{10}{12.0}{\familydefault}{\mddefault}{\updefault}$a$}}}}
\put(597,1479){\makebox(0,0)[b]{\smash{{\SetFigFont{10}{12.0}{\familydefault}{\mddefault}{\updefault}$a$}}}}
\put(2937,2582){\makebox(0,0)[b]{\smash{{\SetFigFont{10}{12.0}{\familydefault}{\mddefault}{\updefault}$a$}}}}
\put(2922,1989){\makebox(0,0)[b]{\smash{{\SetFigFont{10}{12.0}{\familydefault}{\mddefault}{\updefault}$b$}}}}
\put(3957,1494){\makebox(0,0)[b]{\smash{{\SetFigFont{10}{12.0}{\familydefault}{\mddefault}{\updefault}$b$}}}}
\put(1984,1434){\makebox(0,0)[b]{\smash{{\SetFigFont{10}{12.0}{\familydefault}{\mddefault}{\updefault}$a$}}}}
\put(4490,917){\makebox(0,0)[b]{\smash{{\SetFigFont{10}{12.0}{\familydefault}{\mddefault}{\updefault}$b$}}}}
\put(2892,894){\makebox(0,0)[b]{\smash{{\SetFigFont{10}{12.0}{\familydefault}{\mddefault}{\updefault}$b$}}}}
\put(1317,894){\makebox(0,0)[b]{\smash{{\SetFigFont{10}{12.0}{\familydefault}{\mddefault}{\updefault}$b$}}}}
\put(5322,1479){\makebox(0,0)[b]{\smash{{\SetFigFont{10}{12.0}{\familydefault}{\mddefault}{\updefault}$b$}}}}
\put(5322.000,1233.643){\arc{340.715}{2.1274}{7.2974}}
\blacken\path(5460.781,1202.668)(5412.000,1089.000)(5508.533,1166.340)(5460.781,1202.668)
\end{picture}
}
\end{center}
\caption{Quotient DFA $\cA$ of concatenation of partially monotonic DFA's $\cA_3$ and $\cA_4$.} 
\label{fig:concatpm}
\end{figure}

However, $\cA$ is not nearly monotonic.  If $1 < 6$, then   $1 < 5$ by $a$,  $2 < 3$ by input $b$, $3 < 4$ by  $b$ again, and $5 < 1$ by input $a$, which is a contradiction. We get a similar contradiction by assuming $5 < 1$. Hence the class of partially monotonic DFA's is not closed under product. 

Lack of closure under star follows since each class contains the language $\{aa\}$, and the star of $\{aa\}$ is not star-free.
\qed
\end{proof}

\section{Syntactic Complexity of Star-Free Languages for $n=3$}\label{sec:sc3}

We now prove that the bound reached by  monotonic automata is an upper bound
for star-free languages with $n=3$.

For $n=3$, there are $3^3=27$ possible transformations. The transformations $[2,3,1]$, $[3,1,2]$ are cycles of length 3; whereas 
$$ [1,3,2], [2,1,1],  [2,1,2], [2,1,3], [2,3,2],  [3,1,1], [3,2,1], [3,3,1], [3,3,2]$$ have cycles of length 2.
This leaves 16 aperiodic transformation as potentially possible for star-free languages. 

We say that two aperiodic transformations $a$ and $b$ \emph{conflict} if $ab$ or $ba$ has a cycle.
Of the 16 transformations, transformations $[1,2,3]$, $[1,1,1]$, $[2,2,2]$, $[3,3,3]$ cannot create any conflicts.
Hence we consider only the remaining 12; their multiplication table is shown in Table~\ref{tab:pairs}, where the bold face entries indicate conflicting pairs.
Since each of the 12 transformations appears as a product of some other transformations, none is essential. 

\begin{table}[ht]
\caption{Conflicting pairs for $n=3$.}
\label{tab:pairs}
\begin{center}
$
\begin{array}{| c || c  | c | c  | c | c | c|c|c|c|c|c|c|c|c|}    
\hline
\  \
& \ 112 \ &\ 113 \ & \ 121 \    & \ 122 \ & \   131 \ &\ 133\ &\ 221 \ &\ 223 \ & \ 233 \ & \ 312\ &\ 322 \ & \ 323 \ \\
\hline\hline
\ 112 \
& \ 111 \ &\ 111 \ & \ 112 \   & \ 112 \ & \   113 &113&222 &222 & 223 &\bf 331 &\bf 332 &\bf 332 \\
\hline
\ 113 \
& \ 112 \ &\ 113 \ & \ 111 \   & \ 112 \ & \  111 & 113 & 221 &223 & 223 & 333 &\bf 332 & 333 \\
\hline
\ 121 \
& \ 111 \ &\ 111 \ & \ 121 \   & \ 121 \ & \  131 & 131 &222 & 222 & \bf 232 & 313 & 323 & 323 \\
\hline
\ 122 \
& \ 111 \ &\ 111 \ & \ 122 \   & \ 122 \ & \  133 &133 & 222 & 222 & 233 & \bf 311 & 322 & 322 \\
\hline
\ 131 \
& \ 121 \ &\ 131 \ & \ 111 \   & \ 121 \ & \  111 & 131 & \bf 212 & \bf 232 & \bf 232 & 333 & 323 & 333 \\
\hline
\ 133 \
& \ 122 \ &\ 133 \ & \ 111 \   & \ 122 \ & \  111 & 133 & \bf 211&  233 & 233 & 333 & 322 & 333 \\
\hline
\ 221 \
& \ 111 \ &\ 111 \ & \ 221 \   & \ 221 \ & \  \bf 331 & \bf 331 & 222 &  222 & \bf 332 & 113 & 223 & 223 \\
\hline
\ 223 \
& \ 112 \ &\ 113 \ & \ 221 \   & \ 222 \ & \  \bf 331 & 333 & 221 &  223 &  333 & 113 & 222 & 223 \\
\hline
\ 233 \
& \ 122 \ &\ 133 \ & \ \bf 211 \   & \ 222 \ & \  \bf 311 & 333 & \bf 211 &  233 &  333 & 133 & 222 & 233* \\
\hline
\ 313 \
& \ \bf 212 \ &\ 313 \ & \ 111 \   & \ \bf 212 \ & \  111 & 313 & 121 &  323 &  323 & 333 & \bf 232 & 333* \\
\hline
\ 322 \
& \ \bf 211 \ &\ \bf 311 \ & \ 122 \   & 222 \ & \  133 & 333 & 122 &  322 & 333 &\bf 311 & 222 & 322* \\
\hline
\ 323 \
& \ \bf 212 \ &\ 313 \ & \ 121 \   & 222 \ & \  131 & 333 & 121 &  323 & 333 & 313 & 222 & 323 \\
\hline
\end{array}
$
\end{center}
\end{table}

Each of the following six  transformations has only one conflict: $$ [1,1,3], [1,2,1],  [1,2,2], [1,3,3], [2,2,3],  [3,2,3].$$
There are also two conflicting triples shown in dotted lines in the conflict graph of Fig.~\ref{fig:conflicts3}:
$( [1,1,3],   [1,2,2],  [3,2,3])$  and
$( [1,2,1],  [1,3,3], [2,2,3])$;
 \begin{figure}[ht]
\begin{center}
\setlength{\unitlength}{0.00043745in}
\begingroup\makeatletter\ifx\SetFigFont\undefined%
\gdef\SetFigFont#1#2#3#4#5{%
  \reset@font\fontsize{#1}{#2pt}%
  \fontfamily{#3}\fontseries{#4}\fontshape{#5}%
  \selectfont}%
\fi\endgroup%
{\renewcommand{\dashlinestretch}{30}
\begin{picture}(7029,2019)(0,-10)
\put(161,461){\makebox(0,0)[lb]{\smash{{\SetFigFont{8}{9.6}{\familydefault}{\mddefault}{\updefault}$322$}}}}
\put(1362,1370){\arc{210}{1.5708}{3.1416}}
\put(1362,1527){\arc{210}{3.1416}{4.7124}}
\put(1857,1527){\arc{210}{4.7124}{6.2832}}
\put(1857,1370){\arc{210}{0}{1.5708}}
\path(1257,1370)(1257,1527)
\path(1362,1632)(1857,1632)
\path(1962,1527)(1962,1370)
\path(1857,1265)(1362,1265)
\put(2630,1370){\arc{210}{1.5708}{3.1416}}
\put(2630,1527){\arc{210}{3.1416}{4.7124}}
\put(3125,1527){\arc{210}{4.7124}{6.2832}}
\put(3125,1370){\arc{210}{0}{1.5708}}
\path(2525,1370)(2525,1527)
\path(2630,1632)(3125,1632)
\path(3230,1527)(3230,1370)
\path(3125,1265)(2630,1265)
\put(3889,1370){\arc{210}{1.5708}{3.1416}}
\put(3889,1527){\arc{210}{3.1416}{4.7124}}
\put(4384,1527){\arc{210}{4.7124}{6.2832}}
\put(4384,1370){\arc{210}{0}{1.5708}}
\path(3784,1370)(3784,1527)
\path(3889,1632)(4384,1632)
\path(4489,1527)(4489,1370)
\path(4384,1265)(3889,1265)
\put(5149,1370){\arc{210}{1.5708}{3.1416}}
\put(5149,1527){\arc{210}{3.1416}{4.7124}}
\put(5644,1527){\arc{210}{4.7124}{6.2832}}
\put(5644,1370){\arc{210}{0}{1.5708}}
\path(5044,1370)(5044,1527)
\path(5149,1632)(5644,1632)
\path(5749,1527)(5749,1370)
\path(5644,1265)(5149,1265)
\put(117,470){\arc{210}{1.5708}{3.1416}}
\put(117,627){\arc{210}{3.1416}{4.7124}}
\put(612,627){\arc{210}{4.7124}{6.2832}}
\put(612,470){\arc{210}{0}{1.5708}}
\path(12,470)(12,627)
\path(117,732)(612,732)
\path(717,627)(717,470)
\path(612,365)(117,365)
\put(1362,462){\arc{210}{1.5708}{3.1416}}
\put(1362,619){\arc{210}{3.1416}{4.7124}}
\put(1857,619){\arc{210}{4.7124}{6.2832}}
\put(1857,462){\arc{210}{0}{1.5708}}
\path(1257,462)(1257,619)
\path(1362,724)(1857,724)
\path(1962,619)(1962,462)
\path(1857,357)(1362,357)
\put(2630,469){\arc{210}{1.5708}{3.1416}}
\put(2630,626){\arc{210}{3.1416}{4.7124}}
\put(3125,626){\arc{210}{4.7124}{6.2832}}
\put(3125,469){\arc{210}{0}{1.5708}}
\path(2525,469)(2525,626)
\path(2630,731)(3125,731)
\path(3230,626)(3230,469)
\path(3125,364)(2630,364)
\put(3889,469){\arc{210}{1.5708}{3.1416}}
\put(3889,626){\arc{210}{3.1416}{4.7124}}
\put(4384,626){\arc{210}{4.7124}{6.2832}}
\put(4384,469){\arc{210}{0}{1.5708}}
\path(3784,469)(3784,626)
\path(3889,731)(4384,731)
\path(4489,626)(4489,469)
\path(4384,364)(3889,364)
\put(5149,469){\arc{210}{1.5708}{3.1416}}
\put(5149,626){\arc{210}{3.1416}{4.7124}}
\put(5644,626){\arc{210}{4.7124}{6.2832}}
\put(5644,469){\arc{210}{0}{1.5708}}
\path(5044,469)(5044,626)
\path(5149,731)(5644,731)
\path(5749,626)(5749,469)
\path(5644,364)(5149,364)
\path(357,1272)(357,732)
\path(1617,1272)(1617,732)
\path(2877,1272)(2877,732)
\path(4137,1272)(4137,732)
\path(5397,1272)(5397,732)
\put(6417,1377){\arc{210}{1.5708}{3.1416}}
\put(6417,1534){\arc{210}{3.1416}{4.7124}}
\put(6912,1534){\arc{210}{4.7124}{6.2832}}
\put(6912,1377){\arc{210}{0}{1.5708}}
\path(6312,1377)(6312,1534)
\path(6417,1639)(6912,1639)
\path(7017,1534)(7017,1377)
\path(6912,1272)(6417,1272)
\put(6410,477){\arc{210}{1.5708}{3.1416}}
\put(6410,634){\arc{210}{3.1416}{4.7124}}
\put(6905,634){\arc{210}{4.7124}{6.2832}}
\put(6905,477){\arc{210}{0}{1.5708}}
\path(6305,477)(6305,634)
\path(6410,739)(6905,739)
\path(7010,634)(7010,477)
\path(6905,372)(6410,372)
\path(6657,1272)(6657,732)
\path(312,372)(312,12)(2922,12)
	(2922,147)(2922,372)
\path(4092,372)(4092,12)(6702,12)
	(6702,147)(6702,372)
\path(1257,552)(717,552)
\path(2517,552)(1977,552)
\path(5037,552)(4497,552)
\path(6297,552)(5757,552)
\dashline{45.000}(2922,1632)(2922,1992)(312,1992)
	(312,1857)(312,1632)
\dashline{45.000}(6702,1632)(6702,1992)(4092,1992)
	(4092,1857)(4092,1632)
\dashline{45.000}(1257,1452)(717,1452)
\dashline{45.000}(2517,1452)(1977,1452)
\dashline{45.000}(5037,1452)(4497,1452)
\dashline{45.000}(6297,1452)(5757,1452)
\put(169,1362){\makebox(0,0)[lb]{\smash{{\SetFigFont{8}{9.6}{\familydefault}{\mddefault}{\updefault}$113$}}}}
\put(1399,1362){\makebox(0,0)[lb]{\smash{{\SetFigFont{8}{9.6}{\familydefault}{\mddefault}{\updefault}$122$}}}}
\put(2674,1362){\makebox(0,0)[lb]{\smash{{\SetFigFont{8}{9.6}{\familydefault}{\mddefault}{\updefault}$323$}}}}
\put(3941,1362){\makebox(0,0)[lb]{\smash{{\SetFigFont{8}{9.6}{\familydefault}{\mddefault}{\updefault}$121$}}}}
\put(1406,453){\makebox(0,0)[lb]{\smash{{\SetFigFont{8}{9.6}{\familydefault}{\mddefault}{\updefault}$313$}}}}
\put(2681,461){\makebox(0,0)[lb]{\smash{{\SetFigFont{8}{9.6}{\familydefault}{\mddefault}{\updefault}$112$}}}}
\put(3941,461){\makebox(0,0)[lb]{\smash{{\SetFigFont{8}{9.6}{\familydefault}{\mddefault}{\updefault}$233$}}}}
\put(5209,461){\makebox(0,0)[lb]{\smash{{\SetFigFont{8}{9.6}{\familydefault}{\mddefault}{\updefault}$221$}}}}
\put(6432,1362){\makebox(0,0)[lb]{\smash{{\SetFigFont{8}{9.6}{\familydefault}{\mddefault}{\updefault}$223$}}}}
\put(6477,462){\makebox(0,0)[lb]{\smash{{\SetFigFont{8}{9.6}{\familydefault}{\mddefault}{\updefault}$131$}}}}
\put(5186,1362){\makebox(0,0)[lb]{\smash{{\SetFigFont{8}{9.6}{\familydefault}{\mddefault}{\updefault}$133$}}}}
\put(117,1370){\arc{210}{1.5708}{3.1416}}
\put(117,1527){\arc{210}{3.1416}{4.7124}}
\put(612,1527){\arc{210}{4.7124}{6.2832}}
\put(612,1370){\arc{210}{0}{1.5708}}
\path(12,1370)(12,1527)
\path(117,1632)(612,1632)
\path(717,1527)(717,1370)
\path(612,1265)(117,1265)
\end{picture}
}
\end{center}
\caption{Conflict graph for $n=3$.} 
\label{fig:conflicts3}
\end{figure}
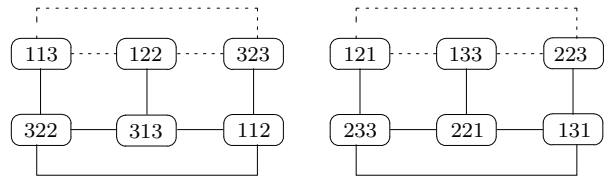
normal lines show conflicting pairs.
We can choose at most two inputs from each triple and at most one from each   conflicting pair. 
Hence there are at most 6 conflict-free transformations from Table~\ref{tab:pairs}, 
for example, $$[1,1,2],[2,2,3],[1,3,3],[1,1,3],[1,2,2],[2,3,3].$$
Adding the identity  and the three constant transformations, we get a total of at most 10.
The inputs shown in Table~\ref{tab:sf3} are conflict-free and generate precisely these 10 transformations.  Hence $\sig(L)\le 10$ for any $L$ with $n=3$.

\section{Conclusions}
\label{sec:con}

We conjecture that the syntactic complexity of languages accepted by the nearly monotonic automata of Definition~\ref{def:nm} meets the upper bound for star-free languages in general: 
\medskip

\noin
{\bf Conjecture 1}
\emph{The syntactic complexity of a star-free language $L$ with $\kappa(L)=n \ge 2$ satisfies $\sig(L) \le {h(n)}$, and this bound is tight.}

\medskip
Our results are summarized in Table~\ref{tab:Summary}. Let $Q = \{1,\ldots,n\}$, and $Q' = Q \setminus \{n\}$. The figures in bold type are tight bounds verified using {\it GAP}~\cite{GAP}, by enumerating aperiodic subsemigroups of $\cT_Q$. The asterisk $\ast$ indicates that the bound is already tight for a smaller alphabet. The last four rows show the values of $f(n) = |M_Q|$, $g(n-1)=|PM_{Q'}|$, the conjectured tight upper bound $h(n) = |PM_{Q'}| + n - 1$, and the weaker upper bound $(n+1)^{n-1}$ (not tight in general).

\begin{table}[ht]
\caption{Syntactic complexity of star-free languages.}
\label{tab:Summary}
\begin{center}
$
\begin{array}{|c||c|c|c|c|c|c|}    
\hline
\ \ |\Sig| \ / \  n \ \ &\ \ \ 1 \ \ \ &\ \ \ 2 \ \ \ & \ \ \  3 \ \ \
& \ \ \ 4 \ \ \ & \ \ \ 5 \ \ \ & \ \ \ 6 \ \ \ \\
\hline \hline

1
&	{\bf 1} 	&	{\bf 1}	&	{\bf 2}	&	{\bf 3}
&	{\bf 5}		&	{\bf 6} \\
\hline

2 
&	\ast		&	{\bf 2}	&	\ {\bf 7} \	
&	{\bf 19}	&	{\bf 62} & ? \\
\hline

3
&	\ast						&	{\bf 3}						&{\bf 9}	
&	{\bf 31}	&	? & ? \\
\hline

4
&	\ast						&	\ast			&\	{\bf 10} 	 \
&	{\bf 34}			&	?	& ? \\
\hline

5
&	\ast						&	\ast						&	\ast					
&	37							&	~125~	& ? \\
\hline

6
&	\ast						&	\ast						&	\ast					
&	40							&	~126~  & ? \\
\hline

7
&	\ast						&	\ast						&	\ast					
&	41							&	~191~  & ? \\
\hline

8
&	\ast						&	\ast						&	\ast					
&	?							&	~195~  & ? \\
\hline

9
&	\ast						&	\ast						&	\ast					
&	?							&	~196~  	& ~1,001~ \\
\hline

10
&	\ast						&	\ast						&	\ast					
&	?							&	?  		& ~1,006~ \\
\hline

11
&	\ast						&	\ast						&	\ast					
&	?							&	?  		& ~1,007~ \\
\hline

\cdots
&								&								&							
&								&    		& \\
\hline

f(n) = |M_Q|
&	1				&	3					& 10				&	35	
&	126  & 462 \\
 \hline

 \ g(n-1)= |CM_Q|=|PM_{Q'}| \
 &	-				&	2					& 8					&	38	
 &	192  & 1,002 \\
 \hline

h(n) = |PM_{Q'}| + n - 1 
&	-				&	3					& 10				&	41
&	196  & 1,007 \\
\hline

(n+1)^{n-1}
&	 1				&	3			& 	16					&  	125	
&	1,296	& 16,807 \\
\hline
\end{array}
$
\end{center}
\end{table}
\medskip

\noin
{\bf Acknowledgment}
We thank Zoltan \'Esik and Judit Nagy-Gyorgy for pointing out to us the relationship between aperiodic inputs and Caley's theorem.

\providecommand{\noopsort}[1]{}


\begin{thebibliography}{10}

\bibitem{AV05}
Ananichev, D., Volkov, M.:
\newblock Synchronizing generalized monotonic automata.
\newblock Theoret. Comput. Sci. \textbf{330}(1) (2005)  3--13

\bibitem{Brz10}
Brzozowski, J.:
\newblock Quotient complexity of regular languages.
\newblock J. Autom. Lang. Comb \textbf{15}(1/2) (2010)  71--89

\bibitem{BLY11}
Brzozowski, J., Li, B., Ye, Y.:
\newblock Syntactic complexity of prefix-, suffix-, and bifix-free regular
  languages.
\newblock In Holzer, M., Kutrib, M., Pighizzini, G., eds.: 13th International
  Workshop on Descriptional Complexity of Formal Systems, $($DCFS\/$)$. Volume
  6808 of LNCS.
\newblock Springer, Berlin Heidelberg (2011)  93--106

\bibitem{BrLiu11}
Brzozowski, J., Liu, B.:
\newblock Quotient complexity of star-free languages.
\newblock In D\"om\"osi, P., Szabolcs, I., eds.: 13th International Conference
  on Automata and Formal Languages (AFL).
\newblock Institute of Mathematics and Informatics, College of Ny\'iregyh\'aza,
  Debrecen, Hungary (2011)  138--152

\bibitem{BrYe10}
Brzozowski, J., Ye, Y.:
\newblock Syntactic complexity of ideal and closed languages.
\newblock In Mauri, G., Leporati, A., eds.: 15th International Conference on
  Developments in Language Theory, $($DLT\/$)$. Volume 6795 of LNCS.
\newblock Springer, Berlin Heidelberg (2011)  117--128

\bibitem{Cay89}
Caley, A.:
\newblock A theorem on trees.
\newblock Quart. J. Math. \textbf{23} (1889)  376--378

\bibitem{GAP}
GAP-Group:
\newblock GAP - Groups, Algorithms, Programming - a System for Computational
  Discrete Algebra, http://www.gap-system.org/ (2010)

\bibitem{GoHo92}
Gomes, G., Howie, J.:
\newblock On the ranks of certain semigroups of order-preserving
  transformations.
\newblock Semigroup Forum \textbf{45} (1992)  272--282

\bibitem{HoKo04}
Holzer, M., K\"{o}nig, B.:
\newblock On deterministic finite automata and syntactic monoid size.
\newblock Theoret. Comput. Sci. \textbf{327}(3) (2004)  319--347

\bibitem{HKM11}
Holzer, M., Kutrib, M., Meckel, K.:
\newblock Nondeterministic state complexity of star-free languages.
\newblock In Bouchou-Markhoff, B., Caron, P., Champarnaud, J.M., Maurel, D.,
  eds.: 16th International Conference on Implementation and Application of
  Automata $($CIAA\/$)$. Volume 6807 of LNCS.
\newblock Springer, Berlin Heidelberg (2011)  178--189

\bibitem{How71}
Howie, J.M.:
\newblock Products of idempotents in certain semigroups of transformations.
\newblock Proc. Edinburgh Math. Soc. \textbf{17}(2) (1971)  223--236

\bibitem{KLS03}
Krawetz, B., Lawrence, J., Shallit, J.:
\newblock State complexity and the monoid of transformations of a finite set
  (2003) {\tt http://arxiv.org/abs/math/0306416v1}.

\bibitem{Mas70}
Maslov, A.N.:
\newblock Estimates of the number of states of finite automata.
\newblock Dokl. Akad. Nauk SSSR \textbf{194} (1970)  1266--1268 (Russian).
  English translation: Soviet Math. Dokl. {\bf 11} (1970), 1373--1375.

\bibitem{McNP71}
McNaughton, R., Papert, S.A.:
\newblock Counter-Free Automata. Volume~65 of M.I.T. research monograph.
\newblock The MIT Press (1971)

\bibitem{Myh57}
Myhill, J.:
\newblock Finite automata and representation of events.
\newblock Wright Air Development Center Technical Report \textbf{57--624}
  (1957)

\bibitem{Ner58}
Nerode, A.:
\newblock Linear automaton transformations.
\newblock Proc. Amer. Math. Soc. \textbf{9} (1958)  541--544

\bibitem{RozSal97}
Rozenberg, G., Salomaa, A., eds.:
\newblock Handbook of Formal Languages, vol. 1: Word, Language, Grammar.
\newblock Springer, New York, NY, USA (1997)

\bibitem{Sch65}
Sch\"utzenberger, M.:
\newblock On finite monoids having only trivial subgroups.
\newblock Inform. and Control \textbf{8} (1965)  190--194

\bibitem{Sho95}
Shor, P.W.:
\newblock A new proof of \mbox{Cayley's} formula for counting labeled trees.
\newblock J. Combin. Theory Ser. A \textbf{71}(1) (1995)  154--158

\bibitem{Yu01}
Yu, S.:
\newblock State complexity of regular languages.
\newblock J. Autom. Lang. Comb. \textbf{6} (2001)  221--234

\end{thebibliography}

\end{document}